\documentclass[review]{elsarticle}

\usepackage{lineno,hyperref}
\modulolinenumbers[5]

\journal{Journal of \LaTeX\ Templates}
\usepackage[utf8x]{inputenc}
\DeclareMathAlphabet\mbi{OML}{cmm}{b}{it}
\DeclareSymbolFont{boldsymbols}{OMS}{cmsy}{b}{n}
\DeclareSymbolFontAlphabet{\mathbfcal}{boldsymbols}

\usepackage{geometry}
\usepackage{amsthm}
\usepackage{amsmath}
\usepackage{amsbsy}
\usepackage{amsfonts}
\usepackage{dsfont}
\usepackage{mathrsfs}
\usepackage{amssymb}
\everymath{\displaystyle}
\newtheorem{thm}{Theorem}[section]

\theoremstyle{definition}

\theoremstyle{Remarque}

\theoremstyle{plain}

\usepackage{graphicx}
\usepackage{rotating}
\usepackage{marvosym}
\usepackage{textcomp}
\usepackage {setspace}
\usepackage{fancyhdr}
\usepackage{hyphenat}
\usepackage{multirow}
\usepackage{pdfpages}
\usepackage{pdftexcmds}
\usepackage{subfigure}
\usepackage{booktabs,caption,fixltx2e}
\usepackage[flushleft]{threeparttable}
\usepackage{latexsym}
\usepackage{lmodern}
\usepackage{a4wide}
\usepackage{lscape}
\usepackage{rotating}
\usepackage[flushleft]{threeparttable}
\DeclareMathOperator*{\argmax}{arg\,max}
\DeclareMathOperator*{\argmin}{arg\,min}
\DeclareMathOperator*{\tr}{tr}

\biboptions{authoryear}
\usepackage{filecontents}

\begin{document}

\begin{frontmatter}

\title{Fixed effects selection in the linear mixed-effects model using adaptive ridge procedure for $L_0$ penalty performance}


\author[mymainaddress,mysecondaryaddress]{Eric Houngla Adjakossa\corref{correspondingauthor}}
\cortext[correspondingauthor]{Corresponding author}
\ead{ericadjakossah@gmail.com}

\author[mysecondaryaddress]{Gregory Nuel}
\ead{Gregory.Nuel@math.cnrs.fr}

\address[mymainaddress]{International Chair in Mathematical Physics and Applications (ICMPA-UNESCO Chair) /University of Abomey-Calavi, 072 B.P. 50 Cotonou, Republic of Benin}
\address[mysecondaryaddress]{ Laboratoire de Probabilit\'es et Mod\`eles Al\'eatoires /Universit\'e Pierre et Marie Curie, Case courrier 188 - 4, Place Jussieu 75252 Paris cedex 05 France}

\begin{abstract}
This paper is concerned with the selection of fixed effects along with the estimation of fixed effects, random effects and variance components in the linear mixed-effects model. We introduce a selection procedure based on an adaptive ridge (AR) penalty of the profiled likelihood, where the covariance matrix of the random effects is Cholesky factorized. This selection procedure is intended to both low and high-dimensional settings where the number of fixed effects is allowed to grow exponentially with the total sample size, yielding technical difficulties due to the non-convex optimization problem induced by $L_0$ penalties. Through extensive simulation studies, the procedure is compared to the LASSO selection  and appears to enjoy the model selection consistency as well as the estimation consistency.
\end{abstract}


\end{frontmatter}


\section{Introduction}
During the last two decades, selection procedures in the linear mixed-effects model have been an active research topic due to the appealing features of the model and the advent of modern technologies facilitating the collection of many variables in scientific studies. Many of these variables are typically included in the full model at the initial stage of modeling to reduce model approximation error, and due to the complexity of the mixed-effects models, inferences and interpretations of the estimated models become challenging as the dimension of fixed or random effects increases~\citep{fan2012variable}. The selection of important fixed or random effects has thus become a fundamental problem in the analysis of grouped data using mixed-effects models, especially in the high-dimensional settings where the fixed or the random effects vector dimension is allowed to grow exponentially with the sample size.

Generally, model selection procedures can be viewed as covering three main approaches: the hypothesis testing procedures, the regularization procedures and other procedures which include the Bayesian selection methods.
The testing procedures include the ordinary hypothesis tests and the selection methods based on generalized information criteria. Examples of using testing hypothesis for models selection in the mixed-effects model context include Lin's works~\citep{lin1997variance} who proposed a simple global variance component tests, which are locally asymptotically most precise and are robust in the sense that no assumption about the parametric form of the random effects is made. Despite their global form expressions which require only the fitting of conventional generalized linear models, the critical values of the global test statistics, which are based on large sample theory, are less accurate when the number of levels of each random effect is small, e.g. less than 15. Edwards and his co-workers \citep{edwards2008r2} extended the traditional coefficient of determination $R^2$ for the linear mixed-effects model $y\sim\mathcal{N}(X\beta,\Sigma=Z\Gamma Z^\top+\sigma^2I_N)$, where they introduced a statistic $R^2_\beta=(q-1)\nu^{-1}F(\widehat{\beta},\widehat{\Sigma})/\left[1+(q-1)\nu^{-1}F(\widehat{\beta},\widehat{\Sigma})\right]$, with $F(\widehat{\beta},\widehat{\Sigma})=(C\widehat{\beta})^\top\left[C(X^\top\widehat{\Sigma}^{-1}X)^{-1}C^\top\right]^{-1}C\widehat{\beta}/\text{rank}(C)$, $\nu=N-\text{rank}(X)=N-q$, $C=\left[\boldsymbol{0}_{(q-1)\times 1}I_{q-1}\right]$ of rank $q-1$, in testing $H_0:C\beta=\boldsymbol{0}$. $R^2_\beta$ measures the multivariate association between the repeated outcomes and the fixed effects in the context of longitudinal data analysis. This $R^2_\beta$ statistic arises as a $1-1$ function of an appropriate $F$ statistic (i.e., $F(\widehat{\beta},\widehat{\Sigma})$) for testing all the fixed effects, except the intercept. More precisely, $R^2_\beta$ compares the full model with a null model having no fixed effect except typically the intercept. $R^2_\beta$ is then generalized to define a partial $R^2$ statistic for marginal fixed effects of all sorts. Although this testing-based selection procedure of fixed effects is very useful, one of its major drawback is that the choice of the denominator of $R^2_\beta$ clearly affects the rate of convergence as $N\rightarrow\infty$, and may change the parameter being estimated. Examples of $R^2$-based selection of fixed effects in the mixed-effects model include~\cite{snijders1994modeled}, \cite{xu2003measuring}, \cite{kramer2005r} and references therein, where generally,  too much restrictions are made on the random effects covariance matrix, and clearly may not be appropriate for a wide range of data analysis. 

Based on testing procedures, a stepwise procedure can be constructed for selecting important fixed or random effects using generalized information criteria~\citep{pu2006selecting,fan2012variable}, which are a generalization of Akaike's information criterion (AIC)~\citep{akaike1974new} and the Bayesian information criterion (BIC)~\citep{schwarz1978estimating}. \cite{pu2006selecting} extended the Generalized information criterion (GIC) proposed by~\cite{rao1989strongly} in order to construct a procedure for selecting fixed and random effects in the linear mixed-effects model. Here, following~\cite{shao1997asymptotic}, the asymptotic behavior of the extended GIC method for selecting fixed effects is studied, and the results from simulations show that if the signal-to-noise ratio is moderate or high, the percentages of choosing the correct fixed effects by the GIC procedure are close to one for finite samples. Another examples of GIC like selection procedures include those proposed by~\cite{mallows1973some}, \cite{hannan1979determination} and ~\cite{bozdogan1987model}. These strategies suggest a unified approach for choosing a parameters vector $\beta$ that maximizes the penalized likelihood
\begin{equation}
n^{-1}\ell_n(\beta|y)-\sum_{j=1}^p p_\lambda(|\beta_j|)
\end{equation}
which arose from the Kullback-Leibler (KL) divergence $-\ell_n(\widehat{\beta}|y)+\lambda\|\beta\|_0$ of the fitted model from the true model~\citep{akaike1973information}, where $\ell_n(.|y)$ is the log-likelihood function, $\widehat{\beta}$ is the maximum likelihood estimator of $\beta$ and $p_\lambda$ is the penalty function indexed by the regularization parameter $\lambda\geqslant 0$. The $L_0$-norm $\|\beta\|_0$ of $\beta$ counts the number of non-vanishing components ($\beta_j\neq 0$) in $\beta$, and arises naturally in many classical model selection methods. Although involving a nice interpretation of the best subset selection, and admitting good sampling properties~\citep{barron1999risk}, its computation is infeasible in high dimensional statistical endeavors~\citep{fan2010selective}. Other penalty functions are regularly used in the literature. A natural generalization of $L_0$ penalty is the so-called bridge (noted $L_{q}$) penalty in~\cite{frank1993statistical}, where $p_\lambda(|\beta|)=\lambda\|\beta\|_{L_q}^q$ for $0<q\leqslant 2$. The $L_{q}$ penalty encompasses $L_0$, $L_1$ - LASSO~\citep{tibshirani1996regression} - and $L_2$ (ridge) penalties. Since none of the $L_q$ penalties satisfies all the required properties (sparsity, approximate unbiasedness and continuity, see~\cite{fan2001variable} for more details) for their resulting parameter estimators, other penalties including SCAD~\citep{fan1997comments,fan2001variable} and MCP~\citep{zhang2007penalized} are introduced in the literature.

Testing-based stepwise selection procedures, where $\lambda$ is fixed ($\lambda=1$ for AIC, $\lambda=\log(n)/2$ for BIC, $\lambda=\log(\log n)/2$ for HQIC~\citep{hannan1979determination}, $\lambda=(\log n+1)/2$ for CAIC~\citep{bozdogan1987model}, for example) are computationally expensive for high-dimensional settings and ignore stochastic errors in the variable selection process~\citep{fan2001variable}. Another severe drawback is their lack of stability~\citep{breiman1996heuristics}. Penalized likelihood approaches using data-driven choice of $\lambda$ are generally preferred to handle the high-dimensional selection problem. These methods are referred to as regularization procedures. Here, the regularization parameter $\lambda$ (also called the tuning parameter) is generally estimated using cross-validation methods~\citep{breiman1995better, tibshirani1996regression, fu1998penalized, fan2001variable}. \cite{ibrahim2011fixed} consider the selection of both fixed and random effects in a general class of mixed effects models by optimizing the penalized likelihood criterion $Q_\lambda(\theta|\theta_{\text{old}})=Q(\theta|\theta_{\text{old}})-n\sum_{j=1}^p p_\lambda(|\beta_j|)$ using the EM algorithm~\citep{dempster1977maximum}, where $\theta$ is the vector of all the unknown parameters, $Q(.|\theta_{\text{old}})$ is the resulting function of the EM algorithm E-step, and $p_\lambda$ is either the SCAD or the adaptive lasso - ALASSO - \citep{zou2006adaptive} penalty. They approximate the integral in $Q(\theta|\theta_{\text{old}})$ by using a Markov chain Monte Carlo and introduce the $\text{IC}_Q(\lambda)=-2Q(\widehat{\theta}_\lambda|\widehat{\theta}_0)+c_n(\widehat{\theta}_\lambda)$ statistic~\citep{ibrahim2008model} for selecting the regularization parameter, where $\widehat{\theta}_0=\argmax_\theta\ell(\theta)$ is the unpenalized maximum likelihood estimate and $c_n(\theta)$ a function of the data and fitted model.

\cite{fan2012variable} introduce a class of variable selection methods for fixed effects using a penalized profiled likelihood method, where the unknown covariance matrix of the random effects is replaced with a suitable proxy matrix. Here, the general idea used is as follows. After writing the joint density $f(y,\gamma)$ of the response variable $y$ and the random effects vector $\gamma$, they consider the penalized profiled likelihood $L_n(\beta,\widehat{\gamma}(\beta))-n\sum_{j=1}^{d_n} p_\lambda(|\beta_j|)$, where $L_n(\beta,\widehat{\gamma}(\beta))=f(y,\widehat{\gamma}(\beta))$ with $\widehat{\gamma}(\beta)$ the empirical Bayes estimate of $\gamma$~\citep{harville1977maximum} in which the covariance matrix of $\gamma$ is replaced with the proxy matrix. $d_n$ may increase with the sample size $n$ and $p_\lambda$ is a concave penalty function (SCAD and LASSO in their simulation and application section). Although the proxy matrix may be different from the true one, it may still yield correct model selection results at the cost of some additional bias~\citep{fan2012variable}.  \cite{schelldorfer2011estimation} deal with theoretical and computational aspects for high-dimensional selection of fixed effects in the linear mixed-effects model, where the consistency of the estimator is proven along with a non-asymptotic oracle result for the adaptive LASSO estimator, under the assumption that the eigenvalues of $Z^\top Z$ are bounded. Here, $Z$ is the random effects design matrix, and an explicit analytical expression of the regularization parameter $\lambda$ is also given.

Other approches for variable selection in linear mixed-effects models include ``fence'' procedure and Bayesian techniques. \cite{jiang2008fence} introduce a class of strategies known as a fence methods intended for variable selection in both linear and generalized linear mixed-effects models. The fence strategy is based on a measure of lack-of-fit that is a quantity $Q_M=Q_M(y,\theta_M)$, where $y$ is the response variable, $M$ indicating a candidate model for the selection and $\theta_M$ denotes the vector of parameters under $M$. Counting among the typically rare and ad-hoc selection procedures, the fence method is computationally very demanding, particularly because it involves the estimation of the standard deviation of the difference of lack-of-fit measures. For more details on these ad-hoc selection procedures, see the nice review paper of~\cite{muller2013model} and references therein. Bayesian model selection requires to assign a prior distribution over the model parameters and compute the posterior probabilities of each of them. These computations can be difficult so are usually carried out by applying sophisticated Markov Chain Monte Carlo (MCMC) algorithms~\citep{muller2013model}. Exemples of Bayesian model selection include~\cite{chen2003random} and~\cite{saville2009testing} who point out that these kinds of MCMC methods are generally time consuming to implement, requiring special software and depend on subjective choice of the hyperparameters in the priors.

In this paper, we discuss the selection of fixed effects in the linear mixed-effects model using an adaptive ridge (AR) penalty of the profiled log-likelihood, where the random effect covariance matrix is Cholesky factorized for solving a preliminary penalized least square problem. The profiled likelihood is calculated by slightly modifying the approach proposed by~\cite{bates2007lme4}. The weights matrix of the AR procedure introduced here is updated in such a way that the procedure converges toward selection with $L_0$ penalty, as have done~\cite{frommlet2016adaptive}. The present selection strategy is intended to both low and high-dimensional settings.

The rest of the article is organized as follows. The profiled log-likelihood is introduced in Section~2. In Section~3, we present the weighted ridge procedure, and Section~4 presents the simulation studies.
\section{profiled log-likelihood for the linear mixed-effects model}
In this section, we consider the classical linear mixed-effect model setting where the number of observations $n$ is larger than the number of covariates $p$. By slightly modifying the approach introduced by~\cite{bates2007lme4}, we calculate the profiled likelihood function.

\subsection{Model and notations}
 We consider the linear mixed-effects model in which the residual terms are homoscedastic and independent of the random effects as follows.
 
 \begin{equation}\label{eq:model11}
 \mathcal{Y}=X\beta+Z\gamma+\varepsilon,
 \end{equation}
 \begin{equation}\label{eq:model12}
\gamma\sim\mathcal{N}(\boldsymbol{0},\Gamma),\quad\varepsilon\sim\mathcal{N}(\boldsymbol{0},\sigma^2I_n)\quad\text{and }\gamma\perp\varepsilon,
 \end{equation}
where $\mathcal{Y}$ is the random response variable whose observed value is $y_{\text{obs}}\in\mathbb{R}^n$, $\gamma\in\mathbb{R}^q$ is the unobserved random effects vector with covariance matrix $\Gamma$, $\varepsilon$ is the residual term, $\beta\in\mathbb{R}^p$ is the fixed effects vector, and $X$ and $Z$ are the fixed and random effects related design matrices of dimensions $n\times p$ and $n\times q$, respectively. $\sigma^2I_n$ is the covariance matrix of $\varepsilon$ with $\sigma > 0$ and $I_n$ the $n\times n$ identity matrix.

As a variance-covariance matrix, $\Gamma$ must be positive semidefinite. Conveniently, the model is expressed in terms of a relative covariance factor, $\Lambda_\theta$, which is a $q\times q$ matrix, depending on the variance components vector, $\theta$, that generate the symmetric $q\times q$ variance-covariance matrix, $\Gamma$, according to
\begin{equation}
\Gamma=\sigma^2\Lambda_\theta\Lambda_\theta^\top,
\end{equation}
where $\sigma$ is the same scale parameter as in Equation~(\ref{eq:model12}). This factorization of $\Gamma$ yields the existence of a random vector $\mathcal{U}$ such that 
\begin{equation}
\gamma=\Lambda_\theta\mathcal{U},
\end{equation}
with
\begin{equation}
\mathcal{U}\sim\mathcal{N}(\boldsymbol{0},\sigma^2I_q)
\end{equation}
which is called a \textit{spherical random effects}\footnote{\cite{bates2007lme4} argued that the term ``spherical'' is related to the fact that the contours of the $\mathcal{N}(0,\sigma^2I_q)$ probability density are spheres.} variable. The model can therefore be re-expressed as follows
 \begin{equation}\label{eq:model21}
 \mathcal{Y}=X\beta+Z\Lambda_\theta\mathcal{U}+\varepsilon,
 \end{equation}
 \begin{equation}\label{eq:model22}
\mathcal{U}\sim\mathcal{N}(\boldsymbol{0},\sigma^2I_q),\quad\varepsilon\sim\mathcal{N}(\boldsymbol{0},\sigma^2I_n)\quad\text{and }\quad\mathcal{U}\perp\varepsilon.
 \end{equation}
 The parameters of the model are the fixed effects vector $\beta$ and the variance components $\theta$ and $\sigma^2$. This formulation of the linear mixed-effects model allows, not only, the use of a singular matrix $\Lambda_\theta$ which arises in practice, but also for a relatively compact expression for the profiled log-likelihood of $\beta$ and $\theta$, conditional on $y_{\text{obs}}$.

\subsection{profiled likelihood}
The profiled likelihood considered here is expressed through the following theorem.
\begin{thm}\label{theorem1}
Suppose that $y$ is a realization of a random vector $\mathcal{Y}$ satisfying the linear mixed-effects model expressed by Equations~(\ref{eq:model21}) and~(\ref{eq:model22}), where $\beta$, $\theta$ and $\sigma^2$ are the parameters to be estimated. Denoting by $L_\theta$ the matrix such that $L_\theta^\top L_\theta=(Z\Lambda_\theta)^\top Z\Lambda_\theta+\text{I}_q$, $\tilde{u}$ the conditional mean of $\mathcal{U}$ given that $\mathcal{Y}=y$, and $g(\tilde{u})=\|y-X\beta-Z\Lambda_\theta \tilde{u}\|^2+\|\tilde{u}\|^2$, the profiled log-likelihood of $\beta$ and $\theta$ conditional on $y$ is 
\begin{equation}
\tilde{\ell}(\beta,\theta|y)=-\frac{1}{2}\log|L_\theta|^2-\frac{n}{2}\left[1+\log \left(\frac{2\pi g(\tilde{u})}{n}\right)\right].
\end{equation}
\end{thm}

\begin{proof}
Denoting by $u$ a realization of $\mathcal{U}$, the density of $y$ is expressed as

\begin{equation}\nonumber
f(y)=\int_{\mathbb{R}^q} f(y,u)du=\int_{\mathbb{R}^q} f(y|u)f(u)du.
\end{equation}

\begin{equation}
f(y|u)f(u)=(2\pi\sigma^2)^{-(n+q)/2}\exp\left[-\frac{\|y-X\beta-Z\Lambda_\theta u\|^2+\|u\|^2}{2\sigma^2}\right],
\end{equation}
where $\|u\|^2=u^\top u$, with $^\top$ denoting the transpose operator.
\begin{equation}
\|y-X\beta-Z\Lambda_\theta u\|^2+\|u\|^2=\left\|\begin{pmatrix}y-X\beta\\0\end{pmatrix}-\begin{pmatrix}Z\Lambda_\theta\\\text{I}_q\end{pmatrix}u\right\|^2=g(u),
\end{equation}
and solving the penalized least squares problem that is to minimize $g(u)$ over $u$ implies
\begin{equation}\label{eq_gu}
\tilde{u}=\argmin_{u\in\mathbb{R}^q}g(u)\iff \begin{pmatrix}Z\Lambda_\theta\\\text{I}_q\end{pmatrix}^\top\begin{pmatrix}Z\Lambda_\theta\\\text{I}_q\end{pmatrix}\tilde{u}=\begin{pmatrix}Z\Lambda_\theta\\\text{I}_q\end{pmatrix}^\top\begin{pmatrix}y-X\beta\\0\end{pmatrix}.
\end{equation}
Viewed as a function of $u$, $g$ is $C^\infty$ and the first and the second differential of g are

\begin{equation}\nonumber
\text{d}g(u)=2\tr\left\{-(y-X\beta-Z\Lambda_\theta u)^\top Z\Lambda_\theta\text{d}u+u^\top\text{d}u\right\}
\end{equation}
and
\begin{equation}\nonumber
\text{d}^2g(u)=2\tr\left\{(Z\Lambda_\theta)^\top Z\Lambda_\theta\text{d}u\text{d}u^\top+\text{d}u\text{d}u^\top\right\}.
\end{equation}
This implies that

\begin{equation}\nonumber
\frac{\partial g}{\partial u}(u)=-2\left(y-X\beta-Z\Lambda_\theta u\right)^\top Z\Lambda_\theta+2u^\top
\end{equation}
and
\begin{equation}\nonumber
\frac{\partial^2 g}{\partial u\partial u^\top}(u)=2\left(Z\Lambda_\theta\right)^\top Z\Lambda_\theta+2\text{I}_q.
\end{equation}
Then, $g(u)$ can be rewritten as
\begin{equation}
g(u)=g(\tilde{u})+\frac{\partial g}{\partial u}(\tilde{u})(u-\tilde{u})+\frac{1}{2}(u-\tilde{u})^\top\frac{\partial^2 g}{\partial u\partial u^\top}(\tilde{u})(u-\tilde{u}).
\end{equation}
Referring to Equation~(\ref{eq_gu}),

\begin{equation}
\begin{pmatrix}Z\Lambda_\theta\\\text{I}_q\end{pmatrix}\left[\begin{pmatrix}y-X\beta\\0\end{pmatrix}-\begin{pmatrix}Z\Lambda_\theta\\\text{I}_q\end{pmatrix}\tilde{u}\right]=0\implies \frac{\partial g}{\partial u}(\tilde{u})=0.
\end{equation}
Then,

\begin{eqnarray}
g(u)&=&g(\tilde{u})+\frac{1}{2}(u-\tilde{u})^\top\frac{\partial^2 g}{\partial u\partial u^\top}(\tilde{u})(u-\tilde{u})\nonumber\\
&=&g(\tilde{u})+(u-\tilde{u})^\top\left[\left(Z\Lambda_\theta\right)^\top Z\Lambda_\theta+\text{I}_q\right](u-\tilde{u})\nonumber\\
&=&g(\tilde{u})+\|L_\theta(u-\tilde{u})\|^2\text{, with }L_\theta^\top L_\theta=\left(Z\Lambda_\theta\right)^\top Z\Lambda_\theta+\text{I}_q,
\end{eqnarray}
and
\begin{eqnarray}
f(y)&=&(2\pi\sigma^2)^{-(n+q)/2}\int_{\mathbb{R}^q} \exp\left[-\frac{g(\tilde{u})+\|L_\theta(u-\tilde{u})\|^2}{2\sigma^2}\right]du\nonumber\\
&=&(2\pi\sigma^2)^{-(n+q)/2}\exp\left[-\frac{g(\tilde{u})}{2\sigma^2}\right]\int_{\mathbb{R}^q} \exp\left[-\frac{\|L_\theta(u-\tilde{u})\|^2}{2\sigma^2}\right]du.
\end{eqnarray}
$v=L_\theta(u-\tilde{u})\implies du=\frac{1}{|L_\theta|}dv$, and

\begin{eqnarray}
f(y)&=&(2\pi\sigma^2)^{-n/2}|L_\theta|^{-1}\exp\left[-\frac{g(\tilde{u})}{2\sigma^2}\right]\int_{\mathbb{R}^q} (2\pi\sigma^2)^{-q/2}\exp\left[-\frac{\|v\|^2}{2\sigma^2}\right]dv\nonumber\\
&=&(2\pi\sigma^2)^{-n/2}|L_\theta|^{-1}\exp\left[-\frac{g(\tilde{u})}{2\sigma^2}\right]
\end{eqnarray}
The log-likelihood of $\beta$, $\theta$ and $\sigma^2$ conditional on $y$ is
\begin{equation}\label{eq:loglik}
\ell(\beta,\theta,\sigma^2|y)=-\frac{n}{2}\log(2\pi\sigma^2)-\frac{1}{2}\log|L_\theta|^2-\frac{g(\tilde{u})}{2\sigma^2}.
\end{equation}
$\frac{\partial\ell(\beta,\theta,\sigma^2|y)}{\partial\sigma^2}=0\implies\sigma^2=\frac{g(\tilde{u})}{n}$. 
By profiling out $\sigma^2$, the profiled log-likelihood, $\tilde{\ell}(\beta,\theta|y)$, of $\beta$ and $\theta$ conditional on $y$ is expressed through
\begin{equation}\nonumber
-2\tilde{\ell}(\beta,\theta|y)=\log|L_\theta|^2+n\left[1+\log \left(\frac{2\pi g(\tilde{u})}{n}\right)\right].
\end{equation}
\end{proof}

\section{$L_0$ estimator of $\beta$ using iteratively weighted ridge procedure}
The $L_0$ estimator of the fixed-effects vector $\beta$ using iteratively weighted ridge procedure presented here fits both the low and the high-dimensional settings. Theoretically, this selection method may enjoy great stability since it uses neither inverse of $X$ (fixed-effects design matrix) nor inverse of $Z$ (random effects design matrix).
\subsection{Adaptive Ridge penalty for the profiled likelihood}
Let us assume that the true underlying fixed-effects vector $\beta_{\text{true}}$ is sparse in the sense that many of its coefficients are zero. To enforce the sparsity of the estimator of $\beta$, we add a weighted $L_2$ (ridge) penalty for the fixed-effects vector $\beta$ to the profiled log-likelihood function. Thus, we are considering the objective function
\begin{equation}\label{eq:w}
\tilde{\ell}_{\lambda,w}(\beta,\theta)=-2\tilde{\ell}(\beta,\theta|y_{\text{obs}})+\lambda\beta^\top W\beta,
\end{equation}
where $W=\text{diag}(w_1,\dots,w_p)$ is a $p\times p$ diagonal matrix and $\lambda\geqslant 0$ is a regularization parameter. We aim at estimating $\beta$, $\theta$ and $\sigma^2$ by
\begin{equation}
(\tilde{\beta}_{\lambda,w},\tilde{\theta}_{\lambda,w})=\argmin_{\beta,\theta}\tilde{\ell}_{\lambda,w}(\beta,\theta)\quad\text{ and }\quad\widetilde{\sigma^2}_{\lambda,w}=\frac{\|y-X\tilde{\beta}_{\lambda,w}-Z\Lambda_{\tilde{\theta}_{\lambda,w}}\tilde{u}\|^2+\|\tilde{u}\|^2}{n},
\end{equation}
where $\tilde{u}$ is calculated by the penalized least squares algorithm that has been hinted at Equation~(\ref{eq_gu}). The criterion $\tilde{\ell}_{\lambda,w}(\beta,\theta)$ is minimized using one of the constrained optimization functions in R~\citep{R}, to provide the estimators $\tilde{\beta}_{\lambda,w}$, $\tilde{\theta}_{\lambda,w}$ and $\widetilde{\sigma^2}_{\lambda,w}$.

\subsection{Iteratively Weighted ridge procedure}
The $L_0$ penalty for regularization arises naturally in many classical model selection since, indeed, it counts the number of non-vanishing parameters, giving a nice interpretation of the best subset selection and admits nice sampling properties~\citep{barron1999risk}. However, its computation is infeasible in high dimensional settings and clearly argued to be a combinational problem with NP-complexity~\citep{fan2010selective}. Some workarounds are reported in the literature, where the proposed procedure converge toward the $L_0$-penalty based selection. For example, \cite{frommlet2016adaptive} introduced an adaptive ridge procedure that helps to approximate $L_0$-penalty performances. Here, we are using the same procedure where the weight matrix ($W$ defined in Equation~(\ref{eq:w})) diagonal elements $w_j, 1\leqslant j\leqslant p$ are iteratively computed and defined as
\begin{equation}\label{eq:w_j}
w_j^{(k)}=\left[\left|\beta_j^{(k)}\right|^2+\delta^2\right]^{-1},
\end{equation}
as have done~\cite{frommlet2016adaptive}, taking inspiration from~\cite{grandvalet1998least,biihlmann2008discussion, candes2008enhancing, rippe2012visualization} and their simulation results. In Equation~(\ref{eq:w_j}), $k$ identifies the iteration and $\beta_j$ is the $j$th component of $\beta$. The general expression of $w_j^{(k)}$ is $w_j^{(k)}=\left[\left|\beta_j^{(k)}\right|^\tau+\delta^\tau\right]^{\frac{q-2}{\tau}}$, where $q$ precises the norm $\|\cdot\|_{L_q}$ for the penalty, $\delta$ calibrates which effect sizes are considered relevant and $\tau$ determines the quality of the approximation $w_j\beta_j^2\approx\left|\beta_j\right|^q$. In practice, $\delta=10^{-5}$ seems to perform well. For more details, see~\cite{frommlet2016adaptive}. More precisely, the selection procedure performed here is as follows. For a fixed $\lambda$, $\tilde{\beta}_{\lambda,w}$ is initialized at $(1,\dots,1)^\top$, the vector $selection$, say, which identifies the selected fixed effects is initialized at $(1,\dots,1)^\top$ and $W$ is initialized at $\text{diag}(1,\dots,1)$. Then the steps come.
\begin{enumerate}
\item[1)] $\tilde{\beta}_{\lambda,\text{old}}\leftarrow\tilde{\beta}_{\lambda,w}$
\item[2)] perform the optimization $(\tilde{\beta}_{\lambda,w},\tilde{\theta}_{\lambda,w})=\argmin_{\beta,\theta}\tilde{\ell}_{\lambda,w}(\beta,\theta)$, initializing $\beta$ by $\tilde{\beta}_{\lambda,\text{old}}$ and $\theta$ by $\theta_0$. Here, the components of $\theta$ which are variances are initialized by $1$ and those which are not variances are initialized by $0$. Thus, $\theta_0$ components are $0$ or $1$.
\item[3)] $w_j\leftarrow\left(\tilde{\beta}_{\lambda,w,j}^2+\delta^2\right)^{-1}$, where $\tilde{\beta}_{\lambda,w,j}$ is the $j$th component of $\tilde{\beta}_{\lambda,w}$, and $w_j$ is the $j$th element of $W$'s diagonal.
\item[4)] $selection_{\text{old}}\leftarrow selection$ and $selection\leftarrow W\cdot\text{diag}(\tilde{\beta}_{\lambda,w})\cdot\tilde{\beta}_{\lambda,w}$
\item[5)] if $|selection-selection_{\text{old}}|<\text{tol}=10^{-5}$, then the selection can be considered as well performed for $\lambda$. Thus, we choose a new value for $\lambda$. If $|selection-selection_{\text{old}}|\geqslant\text{tol}$, the selection does not perform well and we go to the item 1) by choosing the current $\tilde{\beta}_{\lambda,w}$ as $\tilde{\beta}_{\lambda,\text{old}}$, without changing the value of $\lambda$.
\end{enumerate}
For some $\lambda$ values, the $selection$ vector may contain $0$ or $1$ as components. If it is $1$, the corresponding fixed effect $\beta_j$ is selected, and if it is $0$ the corresponding $\beta_j$ is not selected for $\lambda$.

For the choice of the regularization parameter $\lambda$, we propose to use the Bayesian Information Criterion (BIC) criterion defined as
\begin{equation}
c_{n,\lambda}=-2\ell(\hat{\beta}_{\lambda,\text{sel}},\hat{\theta}_{\lambda,\text{sel}},\widehat{\sigma^2}_{\lambda,\text{sel}}|y_{\text{obs}})+\log(n)\cdot\hat{d}_{\lambda},
\end{equation}
where $\hat{\beta}_{\lambda,\text{sel}},\hat{\theta}_{\lambda,\text{sel}}$ and $\widehat{\sigma^2}_{\lambda,\text{sel}}$ are the ML parameters estimators considering the selected variables. $\hat{d}_{\lambda}=\#\{\hat{\beta}_{\lambda,\text{sel},j}\neq 0: 1\leqslant j\leqslant p\}+\dim(\theta)+1$ is the sum of the number of nonzero fixed-effects and the number of variance components. This form of $\hat{d}_\lambda$ has been suggested by \cite{bates2010lme4} and empirically validated by \cite{schelldorfer2014glmmlasso}.

\section{Simulation studies}
In this section, we assess the performance of our approach using simulated data. We compare the obtained results with those coming from the Lasso implementation in a situation where the number of noise variables are excessive. 

Since we use one of the optimizers available in the R software for minimizing the $\tilde{\ell}_{\lambda,w}(\beta,\theta)$ criterion, for too bigger values of $p$, especially when $n\ll p$, the convergence of the used algorithm (``nlminb'' for example) is hardly or no more reached. Due to this convergence problem, we restrict the simulation studies to the low-dimensional setting where $p>40$ and will focus on this problem in another coming paper with the same theoretical approach.

We restrict ourselves to the case of longitudinal data study with $N$ observations coming from $n$ subjects where each subject $i$ has $n_i$ observations. The ``working" data sets are simulated under the following model.
\begin{equation}
y_i=X_{1i}\beta^*+\gamma_i\mathds{1}_{n_i}+\varepsilon_i, \quad \gamma_i\sim\mathcal{N}(0,\Gamma),\quad \varepsilon_i\sim\mathcal{N}(0,\sigma^2I_{n_i}), \quad \text{for}\quad i=1,\dots,n;
\end{equation}
where $\beta^*\in\mathbb{R}^{p_1}$ is the true fixed effects vector, $X_1$ is a $N\times p_1$ design covariates matrix, $\mathds{1}_{n_i}=(1,\dots,1)\in\mathbb{R}^{n_i}$. $\gamma\perp\varepsilon$ and $\gamma_i\perp\gamma_{i'}$ for $i\neq i'$. $\gamma_i$ is an random intercept for the $i$th subject and $\varepsilon$ is the residual term of the model. 

We suppose that we are following up a sample of subjects where the goal is to evaluate how their weights are influenced by other variables including the age, the sex and the nutrition score ``nscore", and which variables govern this influence. The covariates sex, nscore and age are staked in the model matrix $X_1$ and all other covariates are staked in another $N\times p_2$ model matrix $X_2$ such that $X=(X_1|X_2)$ with $\dim(X)=N\times p$ and $p=p_1+p_2$. The components of the vector (variable) age are randomly sampled from a uniform distribution in $[18,37]$ and the nscore variable is also uniform in $[20,50]$. We would like to fit to the data, the model
\begin{equation}
y_i=X_{i}\beta+\gamma_i\mathds{1}_{n_i}+\varepsilon_i, \quad \gamma_i\sim\mathcal{N}(0,\Gamma),\quad \varepsilon_i\sim\mathcal{N}(0,\sigma^2I_{n_i}), \quad \text{for}\quad i=1,\dots,n;
\end{equation}
where $\beta=(\beta_1^\top,\beta_2^\top)^\top$, with $\beta_1\in\mathbb{R}^{p_1}$ and $\beta_2=\boldsymbol{0}\in\mathbb{R}^{p_2}$. We are then challenging to identify which $\beta$ components are zero. 

For the working data sets, we choose $p_1=4$, $p=54$, $N=300$, $n=90$, $\sigma=1$, $\Gamma=1$ and $\beta^*=(1,-1,-1,1)$. For the fixed-effects $\beta$, we have in fact fifty zeros components and only four components are not zeros. We simulate 100 data sets for which we perform both the lasso selection and the IWR (Iteratively Weighted Ridge) procedure. One hundred values of the regularization parameter $\lambda$ are chosen in $[10^{-2},10^2]$. For each replication, the $\lambda$ value that minimizes the BIC criterion is retained for selecting the significant fixed effects. Denoting by $\beta^{**}=(1,-1,-1,1,0,\dots,0)$, the mean squared error $\text{MSE}=\mathbb{E}\left[\|\hat{\beta}-\beta^{**}\|_2^2\right]$ (with $\hat{\beta}$ coming from lasso or IWR) is computed over the 100 replications. The cardinality of the estimated active set (i.e. $|S(\hat{\beta})|$, with $S(\hat{\beta})=\{k:\hat{\beta}_k\neq 0\}$) is computed as well as the proportion TP of true positive (i.e. the selected set is exactly the true one). We also compute the proportion TPC in which the selected set contains the true one and the proportion ZP in which the true zeros are estimated.
\begin{table}[h!]
\caption{{\bf  Selection performances comparison between LASSO and IWR procedures.}}
 \begin{center}
\begin{tabular}[c]{lcc}
 \toprule
\bf Performance criterion & LASSO & IWR \\
 \midrule 
\multirow{1}{1cm}{MSE}
& $1.200$  & $0.254$  \\
\hline
\multirow{1}{1cm}{$|S(\hat{\beta})|$}
& $4 (1.614)$ & $5 (1.250)$ \\ 
\hline
\multirow{1}{1cm}{TP}
& $0\%$ & $35\%$ \\
\hline
\multirow{1}{1cm}{TPC}
& $16\%$ & $90\%$ \\
\hline
\multirow{1}{1cm}{ZP}
& $98\%(0.028)$ & $98\%(0.023)$ \\
\bottomrule
\end{tabular}
 \end{center}
 \label{tab:comparison}
\end{table} 

The selection results based on the 100 simulated data sets are indeed summarized through five performance criteria that are contained in Table~\ref{tab:comparison}, where the numbers between parentheses are the standard deviations related to the criteria mean values (printed just before these parentheses). Information from Table~\ref{tab:comparison} show that the IWR procedure outperforms the LASSO one. IWR selects the correct model $35\%$ of the time when LASSO has never found it (see TP values in Table~\ref{tab:comparison}). For information, we use the R software package lmmlasso~\citep{schelldorfer2011lmmlasso} to perform LASSO selection. $90\%$ of the time, the model selected by IWR contains the true one against $16\%$ for LASSO (see TPC values from Table~\ref{tab:comparison}). Obviously, the estimations are of better qualities from IWR than from LASSO ($\text{MSE}=0.254$ for IWR, and $\text{MSE}=1.200$ for LASSO). 

The true zero estimation proportions (ZP) are computed in Table~\ref{tab:zp} which contains also their number of occurrence over the simulated data sets. 
\begin{table}[h!]
\caption{{\bf  True zero estimation proportions (ZP) with the number of occurrence over the 100 replications.}}
 \begin{center}
\begin{tabular}[c]{lccccccc|cccccc}
 \toprule
 & \multicolumn{7}{c}{ \bf LASSO} & \multicolumn{6}{c}{ \bf IWR} \\
 \midrule 
ZP & 0.88 & 0.90 & 0.92 & 0.94 & 0.96 & 0.98 & 1 & 0.90 & 0.92 & 0.94 & 0.96 & 0.98 & 1 \\
Number & 2 & 2 & 3 & 6 & 14 & 24 & 49 & 2 & 2 & 8 & 16 & 32 & 40 \\
\bottomrule
\end{tabular}
 \end{center}
 \label{tab:zp}
\end{table} 
The LASSO seems to find more often than IWR all the true zeros ($49\%$ for LASSO against $40\%$ for IWR, in Table~\ref{tab:zp}). This may explain the overfitting behavior of IWR ($|S(\hat{\beta})|=5$ for IWR and $|S(\hat{\beta})|=4$ for LASSO in Table~\ref{tab:comparison}). It therefore seems that LASSO has a stronger shrinkage capability than has IWS which shows in turn a somewhat overfitting behavior than LASSO. Through these simulations studies, it appears that the iteratively weighted aspect of the IWR procedure highers the shrinkage performance of the ordinary Ridge and results in a selection method having better performance than the LASSO.

Like in the case of LASSO, it is possible for IWR to take advantage of a warm start of the algorithm to obtain the full regularization path of the selection problem. 
\begin{figure}[!h]
	\begin{center}
		\includegraphics[width=0.6\textwidth,trim=0 0 0 30,clip]{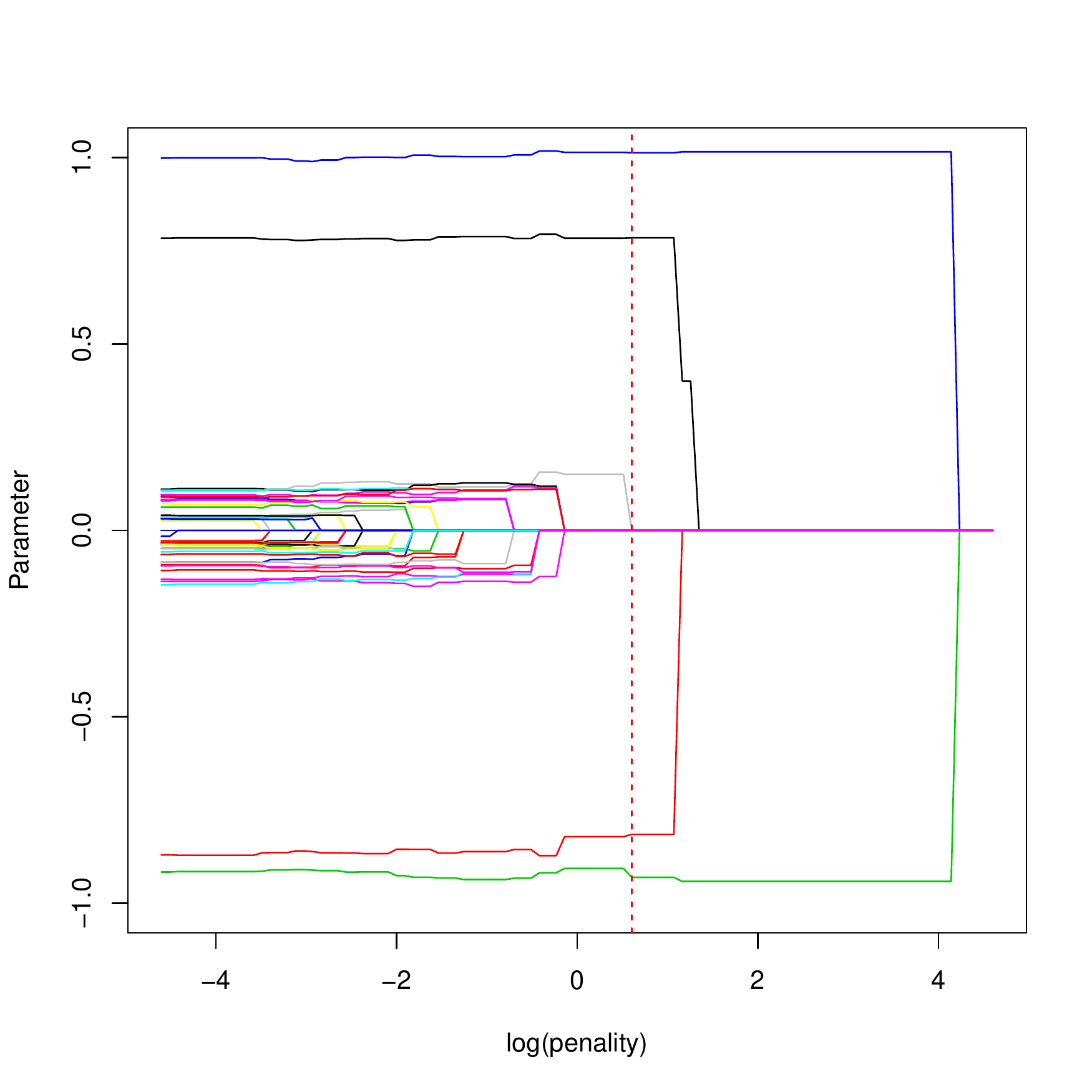}
	\end{center}
	\caption{\textbf{Example of a full regularization path for the iteratively weighted ridge selection procedure with $N=300$, $n=90$, $p=54$, $\sigma=1=\Gamma$. Only the first four components of $\beta^*$ are nonzero. The dataset is simulated using the R software with set.seed(3). The vertical red dashed bar corresponds to the minimum of BIC and shows the selected variables.}}
	\label{fig:fullregupath}
\end{figure}
For instance, Figure~\ref{fig:fullregupath} shows the full regularization path from one of the simulated data. The vertical red dashed bar corresponds to the minimum of BIC and shows the selected variables. On Figure~\ref{fig:fullregupath}, we clearly have $|S(\hat{\beta})|=4$, i.e., four variables selected at the end of the procedure.

\section{Conclusion}
In this paper, we have focused on the fixed-effects selection problem in the linear mixed-effects model. We have introduced an iteratively weighted ridge procedure which enhances the shrinkage performance of the ordinary ridge in order to approximate the performance of the $L_0$ based penalty selection. This selection method is based on an adaptive ridge penalty of the profiled likelihood, where the covariance matrix of the random effects is Cholesky factorized. The procedure fits both the low and the high-dimensional settings and may enjoy a great numeric stability since it needs no use of the inverse of the fixed-effects design matrix or the design matrix of the random-effects. Due the problems of the lack of convergence in the high-dimensional settings using available optimizers, we restrict the simulations studies to the low-dimensional case where we find out that our selection procedure outperforms the LASSO selection. In another forcoming paper, we will focus on this convergence problem in the high-dimensional cases with the same theoretical approach.

\newpage
\section*{References}

\bibliography{\jobname}

\end{document}